\documentclass{ifacconf}

\usepackage{thmtools}
\usepackage{amsmath}
\usepackage{amssymb}
\usepackage{color}
\usepackage{xcolor}
\usepackage{graphicx}
\usepackage{comment}
\usepackage[authoryear]{natbib}
\usepackage{url}
\usepackage[authoryear]{natbib}
\ifdefined\pdfminorversion\pdfminorversion=4\fi
\ifdefined\pdfobjcompresslevel\pdfobjcompresslevel=0\fi

\newcommand{\D}{\mathbb D}

\newcommand{\C}{\mathbb C}
\newcommand{\R}{\mathbb R}

\newcommand{\norm}[1]{\left\lVert#1\right\rVert}

\newtheorem{theorem}{Theorem}
\newtheorem{corollary}{Corollary}

\newtheorem{lemma}{Lemma}
\newtheorem{definition}{Definition}
\newtheorem{assumption}{Assumption}

\newtheorem{remark}{Remark}

\def\beq{\begin{equation}}
\def\eeq{\end{equation}}

\newcommand{\Sclass}{\mathcal{S}}
\newcommand{\NoiseSet}{\mathcal{N}}
\newcommand{\Tset}{\mathcal{T}}
\newenvironment{proof}[1][Proof]
  {\begin{pf*}{#1.}}
  {\unskip\nobreak\hfill$\blacksquare$\end{pf*}}

\newcommand{\Neps}{\mathcal{N}_\varepsilon}

\begin{document}
\begin{frontmatter}
\title{On Sample Complexity Bounds for Frequency-Domain System Identification} 

\author{Arya Honarpisheh}
\author{Mario Sznaier}
\address{ECE Department, Northeastern University, Boston, MA 02115 USA
 (e-mail: honarpisheh.a@northeastern.edu, m.sznaier@northeastern.edu).}\thanks{This work was supported in part by NSF grant CMMI--2208182 and ONR grant N00014-21-1-2431.}

\begin{abstract}

We study the sample complexity of frequency-domain system identification with bounded measurement noise. We focus on interpolatory identification algorithms, which return models consistent with both the data and a priori constraints. For equally spaced samples, we sharpen previous worst-case bounds on the consistency-set radius. Its excess above the noise level decays as $N^{-2}$, where $N$ is the number of samples, and a matching lower bound shows that this rate is tight. We then examine whether bounded stochastic noise allows identification below the noise level. We show that high-probability error bounds can converge to zero. Their convergence rate depends on how strongly the noise distribution is concentrated near the boundary of the noise set. For example, noise uniformly distributed on a disk gives an upper bound of order $(\log N/N)^{2/3}$. Numerical experiments examine these convergence rates and the effect of the noise distribution on the identification error.

\end{abstract}
\begin{keyword}
Frequency domain identification, bounded error identification, error quantification,
set membership estimation, sample complexity, stochastic noise.
\end{keyword}
\end{frontmatter}

\section{Introduction}

Identifying linear time invariant (LTI) systems from noisy data is an important problem in control, signal processing, and machine learning \citep{Ljung2010}. Classical identification techniques typically rely on stochastic noise models and asymptotic arguments to characterize estimation accuracy under specific distributional assumptions \citep{TsiamisEtAl2023}. In contrast, practical applications may offer only limited statistical information about the noise, while bounds on its magnitude may be available. This has motivated a substantial body of work on system identification with disturbances whose distribution is unknown but whose magnitude is bounded \citep{MilaneseTaragna2005}.

Identification with bounded noise naturally leads to a set-theoretic formulation, where the available data define a consistency set of all systems compatible with prior assumptions and the observed measurements. From this viewpoint, no single model can be distinguished without additional criteria, and the fundamental performance limits are dictated by the geometry of this set. This paradigm has a long history in robust and worst-case identification, particularly in the context of $\mathcal{H}_\infty$ system theory and set-membership estimation, where explicit identification error bounds support robust control design \citep{MilaneseVicino1991,HelmickiJacobsonNett1991}.

Recent theoretical advances have revitalized interest in identification with bounded noise by establishing sharper performance bounds and new algorithmic strategies. For instance, \citet{ZengYuOzay2025} derived a $1/N$ lower bound on the best achievable estimation error from $N$ time samples under suitable assumptions on bounded random noise. They also showed that ordinary least squares can remain limited to the slower $1/\sqrt{N}$ rate. Parallel developments in set membership estimation have addressed continuous time MIMO systems using Tustin discretization and convex relaxations \citep{CeroneFossonPirreraRegruto2025}. Additionally, \citet{XuLi2025} analyzed sample complexity and convergence rates when disturbances are bounded by more general convex sets.

In parallel, frequency domain identification offers an attractive alternative to time domain approaches, especially when frequency response data can be directly measured or a model is needed over a selected frequency band \citep{SchoukensPintelonRolain2004}. In this context, \citet{ChenNettFan1995} derived consistency tests, interpolatory algorithms, and error bounds, given prior bounds on stability and gain. Related work considers nonuniform frequency measurements \citep{AkcayGuKhargonekar1994} and combined time and frequency domain data \citep{ParriloEtAl1998}. More recently, \citet{HonarpishehSznaier2025} derived sample complexity bounds for a method that promotes low order models from repeated noisy frequency response measurements.

Our analysis concerns error bounds that apply to every interpolatory identification algorithm. Sections~\ref{sec:notation} and~\ref{sec:problem} introduce the notation and formulate these bounds in terms of consistency sets. In Section~\ref{sec:results}, we strengthen earlier bounds on the radius of the consistency set for zero measurements. Previous upper bounds converge to the noise level at rate $O(N^{-1})$ \citep{ChenGu2002}, or at rate $O(N^{-2})$ to a larger limit of approximately three times the noise level for small noise \citep{ChenNettFan1995}. We obtain $O(N^{-2})$ convergence all the way to the noise level, and a matching lower bound establishes the tight rate $\Theta(N^{-2})$. Section~\ref{sec:stochastic} then considers bounded stochastic noise. We derive upper bounds that can converge to zero and relate their rates to the concentration of the noise distribution near the boundary of the noise set. For uniform disk noise, the upper bound has rate $O((\log N/N)^{2/3})$. A lower bound for uniform disk noise shows why the deterministic quadratic convergence rate need not extend to this setting. The experiments in Section~\ref{sec:numerical} examine both regimes and compare the effect of different noise distributions on the identification error. Section~\ref{sec:conclusion} concludes the paper.

\section{Notation}\label{sec:notation}
The imaginary unit is denoted by $j$.
Let $\D_\rho=\{z\in\C:|z|<\rho\}$ and $\D=\D_1$. The symbols $\overline z$, $\Re z$, and $\Im z$ denote the complex conjugate, real part, and imaginary part of $z$, respectively. We represent the transfer function of a causal discrete-time single-input single-output (SISO) LTI system with impulse response $\{h_k\}_{k\ge0}$ as
\begin{equation}\label{eq:delay}
 H(z)=\sum_{k=0}^{\infty}h_k z^k,\qquad h_k\in\R.
\end{equation}
For $\rho\ge1$, the space $\mathcal H_{\infty,\rho}$ consists of such transfer functions analytic on $\D_\rho$ with finite norm
\[
 \norm{H}_{\infty,\rho}=\sup_{|z|<\rho}|H(z)|.
\]
For $\rho=1$, we write $\mathcal{H}_\infty$ instead of $\mathcal H_{\infty,1}$. For $\rho>1$, $\norm{H}_\infty=\max_{\theta\in[0,2\pi]}|H(e^{j\theta})|$. For $K>0$, define
\begin{equation}\label{eq:boundedclass}
 \mathcal H_{\infty,\rho}^{K}
 =\{H\in\mathcal H_{\infty,\rho}:\norm{H}_{\infty,\rho}\le K\}.
\end{equation}

For a set $\mathcal U$ of transfer functions, define its diameter and radius by
\begin{align*}
 d(\mathcal U)&=\sup_{G,H\in\mathcal U}\norm{G-H}_\infty,\\
 r(\mathcal U)&=\inf_{C\in\mathcal H_\infty}\sup_{H\in\mathcal U}\norm{H-C}_\infty.
 \end{align*}
A minimizer $C$ is called a Chebyshev center of $\mathcal U$ and need not belong to $\mathcal U$ itself.

\section{Problem Formulation}\label{sec:problem}
Consider a discrete-time SISO LTI system $H_0\in\mathcal{S}$, where the set of systems satisfying the prior information is
\begin{equation}\label{eq:prior}
 \Sclass=\mathcal H_{\infty,\rho}^{K},\qquad \rho>1,\quad K>0.
\end{equation}
Every $H\in\Sclass$ has an exponentially decaying impulse response satisfying $|h_k|\le K\rho^{-k}$ for all $k\ge0$. Measurements are available at distinct frequencies $\Theta=\{\theta_n\}_{n=0}^{N-1}$:
\begin{equation}\label{eq:data1}
 y_n=H_0(e^{j\theta_n})+\eta_n,\qquad
 \eta_n\in\NoiseSet_\varepsilon=\{w\in\C:|w|\le\varepsilon\}.
\end{equation}
We allow any noise bound $\varepsilon\ge0$. The frequencies are ordered in $[0,2\pi)$, and their maximum separation on the unit circle is
\begin{equation}\label{eq:gap}
 \Delta=\max_{0\le n<N}(\theta_{n+1}-\theta_n),
 \qquad \theta_N=\theta_0+2\pi.
\end{equation}
For equally spaced samples, $\Delta=2\pi/N$. Here $N$ is the number of frequency samples on the full unit circle. Since the impulse response is real, $H(e^{j(2\pi-\theta)})=\overline{H(e^{j\theta})}$. Measurements on $[0,\pi]$ can therefore be completed by conjugation. In this case, data and noise at paired frequencies satisfy $y_k=\overline{y_n}$ and $\eta_k=\overline{\eta_n}$ when $\theta_k=2\pi-\theta_n$. At $0$ and $\pi$, the measurements and noise are real. The count $N$ includes conjugate copies.

Given $y=(y_0,\ldots,y_{N-1})$, the consistency set is
\begin{equation}\label{eq:consistency}
 \Tset(y)=\{H\in\Sclass:y_n-H(e^{j\theta_n})\in\NoiseSet_\varepsilon
 \text{ for all }n\}.
\end{equation}
Let $\mathcal Y=\{y:\Tset(y)\ne\varnothing\}$ denote the set of admissible measurements. An identification algorithm $\mathcal A$ returns a system $\mathcal A(y)$ for each $y\in\mathcal Y$. It is consistent, or interpolatory, if $\mathcal A(y)\in\Tset(y)$ for every $y\in\mathcal Y$.

For given measurements $y$, the local worst-case identification error is
\begin{equation}\label{eq:localerror}
 e(\mathcal A,y)=\sup_{H\in\Tset(y)}\norm{H-\mathcal A(y)}_\infty.
\end{equation}
The global worst-case identification error is
\begin{equation}\label{eq:globalerror}
 E(\mathcal A)=\sup_{y\in\mathcal Y}e(\mathcal A,y).
\end{equation}

The radius and diameter of information are defined by
\begin{align}
 \mathcal R(I)&=\sup_{y\in\mathcal Y}r(\Tset(y)),\label{eq:inforadius}\\
 \mathcal D(I)&=\sup_{y\in\mathcal Y}d(\Tset(y)),\label{eq:infodiameter}
\end{align}
where $I$ denotes the information specified by the prior class $\Sclass$, the sampling frequencies $\Theta$, and the noise set $\NoiseSet_\varepsilon$. For any consistent algorithm, the local error satisfies \citep[Section~10.2.3]{SanchezPenaSznaier1998}
\begin{equation}\label{eq:localgeometry}
 r(\Tset(y))\le e(\mathcal A,y)\le d(\Tset(y))\le2r(\Tset(y)).
\end{equation}
Taking the supremum over admissible measurements gives
\begin{equation}\label{eq:globalgeometry}
 \mathcal R(I)\le E(\mathcal A)\le\mathcal D(I)\le2\mathcal R(I).
\end{equation}

In the stochastic setting, $\varepsilon>0$ and the noise vector $\eta$ in \eqref{eq:data1} has a specified joint distribution satisfying $|\eta_n|\le\varepsilon$ almost surely. The measurements $y$ and the consistency set $\Tset(y)$ are then random. Probability is taken over $\eta$, with the true system $H_0\in\Sclass$ fixed.
\begin{definition}\label{def:stochasticdiameter}
For $0<\delta<1$, a deterministic bound $b_N(\delta)\ge0$ is a high-probability diameter bound if
\begin{equation}\label{eq:stochasticdiameter}
 \Pr_\eta\!\left\{d(\Tset(y))\le b_N(\delta)\right\}\ge1-\delta
 \quad\text{for every }H_0\in\Sclass.
\end{equation}
\end{definition}
Since $H_0\in\Tset(y)$ almost surely, \eqref{eq:localerror} and \eqref{eq:localgeometry} imply that $b_N(\delta)$ also bounds the identification error of every consistent algorithm with probability at least $1-\delta$.

Our objective is to relate the number and placement of frequency samples to the accuracy guaranteed by consistent identification algorithms. In the worst case, we seek upper and lower bounds on $E(\mathcal A)$ and determine how they depend on the prior information and the noise bound. In the stochastic case, we derive high-probability upper bounds on the consistency-set diameter and use lower bounds to rule out faster convergence rates. We ask under what conditions these bounds can fall below the noise level and converge to zero, and how their convergence depends on the noise distribution and the confidence level.

\section{Worst-Case Identification Bounds}\label{sec:results}
We first express the identification-error bounds in terms of the consistency set for zero measurements. This reduces the analysis to bounding the norm of systems whose sampled frequency responses are bounded by the noise level. We then derive an upper bound using second-derivative estimates, construct systems that give lower bounds, and combine these results to determine the asymptotic convergence rates of the global identification error.

\subsection{Consistency sets and identification error}
The formulation in Section~\ref{sec:problem} satisfies the following assumptions, which allow us to apply standard bounds on the identification error \citep[Chapter~10]{SanchezPenaSznaier1998}.

\begin{assumption}\label{ass:consistency}
The true system belongs to the consistency set: $H_0\in\Tset(y)$.
\end{assumption}

\begin{assumption}\label{ass:convexity}
The prior set $\Sclass$ and the noise set $\NoiseSet_\varepsilon$ are convex and symmetric about zero.
\end{assumption}

Under these assumptions, the following lemma gives an alternative to \eqref{eq:globalgeometry} based on the consistency set for zero measurements.
\begin{lemma}\label{lem:geometry}
Under Assumptions~\ref{ass:consistency}--\ref{ass:convexity}, every interpolatory identification algorithm satisfies \citep[Lemmas~10.3--10.4]{SanchezPenaSznaier1998}
\begin{equation}\label{eq:algorithm}
 \sup_{H\in\Tset(0)}\norm{H}_\infty
 \le E(\mathcal A)
 \le2\sup_{H\in\Tset(0)}\norm{H}_\infty,
\end{equation}
where the consistency set for zero measurements is
\begin{equation}\label{eq:zeroconsistency}
 \Tset(0)=\{H\in\Sclass:|H(e^{j\theta_n})|\le\varepsilon
                  \text{ for all }n\}.
\end{equation}
Moreover,
\begin{align}\label{eq:diameter}
 r(\Tset(0))&=\sup_{H\in\Tset(0)}\norm{H}_\infty,\nonumber\\
 \mathcal D(I)&=d(\Tset(0))=2r(\Tset(0)).
\end{align}
\end{lemma}

The set $\Tset(0)$ consists of systems whose responses at the sampled frequencies can be canceled by admissible noise. Its radius is therefore the largest $H_\infty$ norm compatible with these sample constraints. Lemma~\ref{lem:geometry} allows us to bound the identification error by studying this single set, without taking the supremum over measurement vectors required by $\mathcal R(I)$. In particular, \eqref{eq:diameter} identifies $2r(\Tset(0))$ with the diameter of information.

Since the constant system $H(z)=\min\{\varepsilon,K\}$ belongs to $\Tset(0)$, its $H_\infty$ norm gives a lower bound on $r(\Tset(0))$. The prior bounds the $H_\infty$ norm of every system in $\Tset(0)$ by $K$, giving the upper bound. Combining these two observations yields
\begin{equation}\label{eq:trivial}
 \min\{\varepsilon,K\}\le r(\Tset(0))\le K.
\end{equation}
Consequently,
\begin{equation}\label{eq:geometrycases}
\begin{cases}
 r(\Tset(0))=K,\quad\mathcal D(I)=2K, & \varepsilon\ge K,\\[3pt]
 \varepsilon\le r(\Tset(0))\le K, & 0\le\varepsilon<K.
\end{cases}
\end{equation}
The radius of information satisfies $r(\Tset(0))\le\mathcal R(I)\le K$. For $\varepsilon<K$, the remaining question is how much $r(\Tset(0))$ can exceed $\varepsilon$, given the number and placement of frequency samples.

With this characterization, we now derive upper and lower bounds on $r(\Tset(0))$. We will consider a generic configuration for the sampling frequencies and then specialize the results to the equally spaced case. The following results show that, for equally spaced samples and $0<\varepsilon<K$, the upper bound approaches $\varepsilon$ at the rate $O(N^{-2})$. Further, we show via an example that this decay rate is tight.

\subsection{Upper bound on $r(\Tset(0))$}
To bound the increase in response magnitude between samples, we use a Taylor argument similar to that in \citet[proofs of Theorems~1 and~2]{SharifGharaviAlkhansariKhalaj2003}. We compare the maximum response magnitude with a nearby sampled value by rotating the complex frequency response so that its value at a maximizing frequency is real and nonnegative. Its real part then has a maximum at that frequency, so its first derivative is zero. Since this first derivative vanishes, a second derivative bound controls the difference between the maximum and the nearby sampled value by a term proportional to $\Delta^2$, where $\Delta$ is the maximum angular distance between adjacent samples. For equally spaced samples, $\Delta=2\pi/N$, yielding the $O(N^{-2})$ convergence rate. The following lemma provides the required derivative bound.
\begin{lemma}\label{lem:curvature}
For $H\in\Sclass$, let $F(\theta)=H(e^{j\theta})$. Then
\begin{equation}\label{eq:curvature}
 |F''(\theta)|\le \frac{K\rho}{(\rho-1)^2}
 \left(1-\frac{|F(\theta)|^2}{K^2}\right).
\end{equation}
\end{lemma}
\begin{proof}
Lemma~\ref{lem:curvature} follows by applying Ruscheweyh's inequality, as stated in \citet[Theorem~1]{AndersonRovnyak2006}, to the first and second derivatives. The proof is given in Appendix~\ref{app:proofs}.
\end{proof}

For $0<\varepsilon<K$, we apply the lemma on an interval where $|F(\theta)|\ge\varepsilon$. The factor $1-|F(\theta)|^2/K^2$ is then at most $1-(\varepsilon/K)^2<1$. The following theorem combines this estimate with the prior bound and a separate treatment of the noiseless case using Blaschke factorization.

\begin{theorem}\label{thm:upper}
Under \eqref{eq:prior}--\eqref{eq:gap}, $r(\Tset(0))\le U_\Delta$, where
\begin{equation}\label{eq:uppergap}
 U_\Delta=
 \begin{cases}
 K, & \varepsilon\ge K,\\[3pt]
 \min\left\{K,\varepsilon+
       \dfrac{\rho\Delta^2(K^2-\varepsilon^2)}{8K(\rho-1)^2}\right\},
       & 0<\varepsilon<K,\\[5pt]
 K\left(\dfrac{2\rho}{1+\rho^2}\right)^N, & \varepsilon=0.
 \end{cases}
\end{equation}
In particular, for equally spaced samples, $r(\Tset(0))\le U_N$, where
\begin{equation}\label{eq:upperN}
 U_N=
 \begin{cases}
 K, & \varepsilon\ge K,\\[3pt]
 \min\left\{K,\varepsilon+
       \dfrac{\pi^2\rho(K^2-\varepsilon^2)}{2K(\rho-1)^2N^2}\right\},
       & 0<\varepsilon<K,\\[5pt]
 \dfrac{2K\rho^N}{1+\rho^{2N}}, & \varepsilon=0.
 \end{cases}
\end{equation}
\end{theorem}
\begin{proof}
Case 1: $\varepsilon\ge K$. Every $H\in\Sclass$ satisfies the measurement constraints for $y=0$, so $\Tset(0)=\Sclass$. The constant system $H(z)=K$ attains the upper bound, giving $r(\Tset(0))=K$.

Case 2: $0<\varepsilon<K$. Let $H\in\Tset(0)$ and $R=\norm{H}_\infty$. If $R\le\varepsilon$, the result follows immediately. Otherwise, choose a frequency $\theta_*$ at which $|F|$ attains its maximum and write $F(\theta_*)=H(e^{j\theta_*})=R e^{j\alpha}$. The distance to the nearest sampled frequency, measured along the unit circle, is $d\le\Delta/2$. Choose $\sigma\in\{-1,1\}$ so that $\theta_*+\sigma d$ is this sampled frequency, modulo $2\pi$, and define
\[
 g(s)=\Re\{e^{-j\alpha}F(\theta_*+\sigma s)\},
 \quad 0\le s\le d,\quad \sigma\in\{-1,1\}.
\]
This real-valued function can be defined by the same formula for all $s\in\R$. It satisfies $g(s)\le|F(\theta_*+\sigma s)|\le R$, with equality at $s=0$. Therefore $g'(0)=0$ and $g''(0)\le0$. At the sampled frequency, $g(d)\le\varepsilon$. Since $g(0)=R>\varepsilon$ and $g(d)\le\varepsilon$, continuity guarantees a smallest $\tau\in(0,d]$ such that $g(\tau)=\varepsilon$. For $0\le s\le\tau$, we have $g(s)\ge\varepsilon$.

Differentiating twice gives $g''(s)=\Re\{e^{-j\alpha}F''(\theta_*+\sigma s)\}$. Applying Lemma~\ref{lem:curvature} to $H$ at $\theta_*+\sigma s$ and using $|F(\theta_*+\sigma s)|\ge g(s)\ge\varepsilon$ yields
\begin{equation}\label{eq:crossingcurv}
 g''(s)\ge-|F''(\theta_*+\sigma s)|
 \ge-\frac{K\rho}{(\rho-1)^2}\left(1-\frac{\varepsilon^2}{K^2}\right).
\end{equation}
By Taylor's theorem, there exists $\xi\in(0,\tau)$ such that, using $g(0)=R$, $g'(0)=0$, and \eqref{eq:crossingcurv},
\begin{align*}
 \varepsilon=g(\tau)
 &=g(0)+\tau g'(0)+\frac{\tau^2}{2}g''(\xi)\\
 &\ge R-\frac{K\rho\tau^2}{2(\rho-1)^2}
                    \left(1-\frac{\varepsilon^2}{K^2}\right).
\end{align*}
Rearranging and using $\tau\le d\le\Delta/2$ gives
\[
 R\le\varepsilon+\frac{\rho\Delta^2(K^2-\varepsilon^2)}{8K(\rho-1)^2}.
\]
Combining this with the prior bound $R\le K$ yields
\[
 \norm{H}_\infty\le\min\left\{K,\varepsilon+
 \frac{\rho\Delta^2(K^2-\varepsilon^2)}{8K(\rho-1)^2}\right\}.
\]
This holds for every $H\in\Tset(0)$. Taking the supremum over this set gives \eqref{eq:uppergap}, since $r(\Tset(0))=\sup_{H\in\Tset(0)}\norm{H}_\infty$. Substituting $\Delta=2\pi/N$ gives the second branch of \eqref{eq:upperN}.

Case 3: $\varepsilon=0$. Every $H\in\Tset(0)$ vanishes at $z_n=e^{j\theta_n}$. Define the Blaschke product
\begin{equation}\label{eq:blaschke-product}
 B_\Theta(z)=\prod_{n=0}^{N-1}\frac{\rho(z-z_n)}{\rho^2-\overline{z_n}z}.
\end{equation}
To explain the factorization, write $f_0=H/K$ and denote the $n$th factor in \eqref{eq:blaschke-product} by $b_n$. Then $|f_0|\le1$ on $\D_\rho$ and $f_0(z_0)=0$. The map $b_0$ sends $\D_\rho$ onto the unit disk and $z_0$ to zero. Its inverse is
\[
 b_0^{-1}(w)=\frac{\rho^2w+\rho z_0}{\rho+\overline{z_0}w},\qquad w\in\D.
\]
Schwarz's lemma states that if $f$ is analytic on $\D$, $|f|\le1$, and $f(0)=0$, then $|f(w)|\le|w|$ for $w\in\D$. The function $f=f_0\circ b_0^{-1}$ satisfies these conditions. Taking $w=b_0(z)$ gives
\[
 |f_0(z)|\le |b_0(z)|,\qquad z\in\D_\rho.
\]
Consequently, $f_0/b_0$ is analytic on $\D_\rho$ and bounded by one. Repeating this process for the remaining zeros gives an analytic function $Q$ on $\D_\rho$ satisfying
\[
 Q=\frac{H}{K b_0\cdots b_{N-1}},\qquad
 \norm{Q}_{\infty,\rho}\le1.
\]
Thus $H=KB_\Theta Q$ and $\norm{H}_\infty\le K\norm{B_\Theta}_\infty$. Each factor satisfies, on $|z|=1$,
\[
 \left|\frac{\rho(z-z_n)}{\rho^2-\overline{z_n}z}\right|
 \le\frac{2\rho}{1+\rho^2},
\]
with equality at $z=-z_n$. Multiplying these inequalities proves the noiseless branch of \eqref{eq:uppergap}.

For equally spaced samples, we can compute the product and its norm explicitly. Rotating the grid by $\phi$ gives
\[
 B_{\Theta+\phi}(e^{j\phi}z)=e^{jN\phi}B_\Theta(z),
\]
so the norm on the unit circle is unchanged. We may therefore take $\theta_n=(2n+1)\pi/N$. The sampling points $z_n$ are the roots of $w^N+1$, so
\begin{equation}\label{eq:uniform-grid-polynomial}
 \prod_{n=0}^{N-1}(w-z_n)=w^N+1,\qquad w\in\mathbb C.
\end{equation}
Taking complex conjugates, substituting $\overline w=\rho^2/z$ for $z\ne0$, and multiplying both sides by $z^N$ gives
\begin{equation}\label{eq:uniform-grid-denominator}
 \prod_{n=0}^{N-1}(\rho^2-\overline{z_n}z)=\rho^{2N}+z^N.
\end{equation}
This identity also holds at $z=0$ as well. In \eqref{eq:blaschke-product}, we use \eqref{eq:uniform-grid-polynomial} with $w=z$ to evaluate the numerator and \eqref{eq:uniform-grid-denominator} to evaluate the denominator. This gives
\begin{equation}\label{eq:blaschke-uniform}
 B_\Theta(z)=\frac{\rho^N(z^N+1)}{\rho^{2N}+z^N}.
\end{equation}
Taking the $H_\infty$ norm of \eqref{eq:blaschke-uniform} on the unit circle gives
\[
 \norm{B_\Theta}_\infty=\frac{2\rho^N}{1+\rho^{2N}},
\]
with the maximum attained at $z^N=1$.
Finally, $\norm{H}_\infty\le K\norm{B_\Theta}_\infty$ gives the noiseless branch of \eqref{eq:upperN}.

\end{proof}

\begin{remark}
The noiseless bounds follow from classical Blaschke factorization. This construction also appears in the identification lower bound of \citet[Theorem~5.1]{ChenNettFan1995}. For equally spaced samples and fixed $0<\varepsilon<K$ and $\rho>1$, \citet[Theorem~4.4]{ChenGu2002} gives the upper bound $\varepsilon+K\rho\pi/((\rho^2-1)N)$ for $r(\Tset(0))$, which converges to $\varepsilon$ at rate $O(N^{-1})$. The diameter bound in \citet[Theorem~5.3]{ChenNettFan1995} gives a radius bound after division by two, since \eqref{eq:diameter} gives $\mathcal D(I)=2r(\Tset(0))$. This radius bound has an $O(N^{-2})$ term but converges to $\varepsilon+2\varepsilon K^2/(K^2+\varepsilon^2)$, approximately $3\varepsilon$ when $\varepsilon/K$ is small. The upper bound on $r(\Tset(0))$ in Theorem~\ref{thm:upper} converges all the way to the noise level $\varepsilon$ at rate $O(N^{-2})$ improving upon previous bounds.
\end{remark}

\subsection{Lower bound on $r(\Tset(0))$}
To show that the rate in Theorem~\ref{thm:upper} is sharp for $0<\varepsilon<K$, we construct an FIR system with a peak at $\theta=0$, halfway between samples. We choose its impulse-response coefficients to satisfy the prior and scale its largest sampled magnitude to $\varepsilon$. Since the response magnitude decreases quadratically near the peak, the nearest samples, at distance $\pi/N$, miss the maximum by a term of order $N^{-2}$. This gives a lower bound on $r(\Tset(0))$ with the same rate.

\begin{theorem}\label{thm:lower}
Under \eqref{eq:prior}--\eqref{eq:gap} with $\varepsilon\ge0$ and $N\ge2$, let the sampling frequencies be
\begin{equation}\label{eq:halfgrid}
 \Theta=\left\{\frac{(2n+1)\pi}{N}:n=0,\ldots,N-1\right\}.
\end{equation}
Then $r(\Tset(0))\ge L_N$, where
\begin{equation}\label{eq:FIRlower}
 L_N=
 \begin{cases}
 K, & \varepsilon\ge K,\\[4pt]
 \displaystyle\frac{\varepsilon}{\sqrt{1-
 \frac{4\varepsilon\rho(K-\varepsilon)}{[K+\varepsilon(\rho-1)]^2}
 \sin^2\!\frac{\pi}{2N}}}, & 0<\varepsilon<K,\\[5pt]
 \dfrac{2K\rho^N}{1+\rho^{2N}}, & \varepsilon=0.
 \end{cases}
\end{equation}
For $0<\varepsilon<K$, as $N\to\infty$,
\begin{equation}\label{eq:lowerasymptotic}
 L_N=\varepsilon+
 \frac{\pi^2\varepsilon^2\rho(K-\varepsilon)}{2[K+\varepsilon(\rho-1)]^2N^2}
 +O(N^{-4}).
\end{equation}
\end{theorem}
\begin{proof}
Case 1: $\varepsilon\ge K$. The constant system $H(z)=K$ is admissible and attains $r(\Tset(0))=K$, proving the first branch of \eqref{eq:FIRlower}.

Case 2: $0<\varepsilon<K$. Set
\begin{equation}\label{eq:tchoice}
 t=\frac{K-\varepsilon}{\varepsilon\rho}>0,
\end{equation}
and construct the FIR system
\begin{equation}\label{eq:FIR}
 D_N=\sqrt{1+t^2+2t\cos(\pi/N)},\qquad
 H_N(z)=\frac{\varepsilon(1+tz)}{D_N}.
\end{equation}
On $|z|=\rho$, the values of $1+tz$ trace a circle centered at $1$ with radius $t\rho$, so the maximum magnitude is $1+t\rho$, attained at $z=\rho$. Thus,
\begin{equation}\label{eq:FIRnorms}
 \norm{H_N}_{\infty,\rho}=\frac{\varepsilon(1+t\rho)}{D_N}=\frac{K}{D_N}.
\end{equation}
The second equality follows by substituting $t$ from \eqref{eq:tchoice}. For $N\ge2$, $\cos(\pi/N)\ge0$, so $D_N\ge1$ and $\norm{H_N}_{\infty,\rho}\le K$, establishing the prior constraint in \eqref{eq:prior}.

Since $|1+te^{j\theta}|^2=1+t^2+2t\cos\theta$, the largest sampled magnitude occurs at the frequencies nearest to $0$. On the grid \eqref{eq:halfgrid}, these are at angular distance $\pi/N$. Thus the normalization by $D_N$ gives
\[
 \max_n|H_N(e^{j\theta_n})|
 =\frac{\varepsilon\sqrt{1+t^2+2t\cos(\pi/N)}}{D_N}
 =\varepsilon.
\]
The system therefore is consistent with the measurement data $y=0$, so $H_N\in\Tset(0)$ and $r(\Tset(0))\ge\norm{H_N}_\infty$. Setting $\rho=1$ in \eqref{eq:FIRnorms}, with $t$ fixed, gives $\norm{H_N}_\infty=\varepsilon(1+t)/D_N$. Since $D_N^2=(1+t)^2-4t\sin^2(\pi/(2N))$, this becomes
\[
 \norm{H_N}_\infty
 =\frac{\varepsilon}{\sqrt{1-\frac{4t}{(1+t)^2}\sin^2(\pi/(2N))}}.
\]
Substituting $t$ from \eqref{eq:tchoice} yields the lower bound in the second branch of \eqref{eq:FIRlower}. To obtain the convergence rate, we have
\[
 L_N-\varepsilon
 =\frac{2\varepsilon t\bigl(1-\cos(\pi/N)\bigr)}
 {D_N(1+t+D_N)}.
\]
As $N\to\infty$, the numerator satisfies $2\varepsilon t(1-\cos(\pi/N))=\varepsilon t\pi^2/N^2+O(N^{-4})$.
For the denominator, $D_N^2=(1+t)^2-t\pi^2/N^2+O(N^{-4})$ and $D_N=1+t+O(N^{-2})$. Consequently,
\[
 \begin{aligned}
 D_N(1+t+D_N)&=2(1+t)^2+O(N^{-2}),\\
 \frac{1}{D_N(1+t+D_N)}&=\frac{1}{2(1+t)^2}+O(N^{-2}).
 \end{aligned}
\]
Multiplying the numerator by this reciprocal gives
\[
 L_N-\varepsilon
 =\frac{\varepsilon t\pi^2}{2(1+t)^2N^2}+O(N^{-4}).
\]
Finally, substituting \eqref{eq:tchoice} gives \eqref{eq:lowerasymptotic}.

Case 3: $\varepsilon=0$. Choose $H_N=KB_\Theta$, with $B_\Theta$ given by \eqref{eq:blaschke-uniform}. This system has real coefficients, satisfies $\norm{H_N}_{\infty,\rho}=K$, and vanishes at every sampling frequency in \eqref{eq:halfgrid}. Hence $H_N\in\Tset(0)$ and
\[
 r(\Tset(0))\ge\norm{H_N}_\infty
 =\frac{2K\rho^N}{1+\rho^{2N}},
\]
which proves the noiseless branch of \eqref{eq:FIRlower}.

\end{proof}

\begin{remark}
For $\varepsilon=0$, Theorems~\ref{thm:upper} and~\ref{thm:lower} give the exact value on the grid \eqref{eq:halfgrid}:
\begin{equation}\label{eq:noiseless}
 r(\Tset(0))=\frac{2K\rho^N}{1+\rho^{2N}}
 =2K\rho^{-N}+O(\rho^{-3N}).
\end{equation}
For positive noise, the Blaschke-product construction in the proof of \citet[Theorem~5.1]{ChenNettFan1995} gives, for the grid \eqref{eq:halfgrid} and $0<\varepsilon<K$,
\begin{equation}\label{eq:blaschkelower}
 r(\Tset(0))\ge
 \frac{K\varepsilon(\rho^{2N}+1)+2K^2\rho^N}
 {K(\rho^{2N}+1)+2\varepsilon\rho^N}.
\end{equation}
The earlier lower bound approaches $\varepsilon$ exponentially fast. Our lower bound in \eqref{eq:lowerasymptotic} approaches $\varepsilon$ at rate $\Theta(N^{-2})$, so it is larger for all sufficiently large $N$. Together with the upper bound in Theorem~\ref{thm:upper}, this establishes that the quadratic convergence rate is tight.
\end{remark}

\subsection{Global identification error}
We now combine the upper and lower bounds on $r(\Tset(0))$ to bound the global worst-case identification error. The following corollary applies to every consistent identification algorithm and includes the noiseless case and both positive-noise regimes.
\begin{corollary}\label{cor:error}
Under \eqref{eq:prior}--\eqref{eq:gap} with $\varepsilon\ge0$, $N\ge2$, and the sampling frequencies in \eqref{eq:halfgrid}, every consistent identification algorithm satisfies
\begin{equation}\label{eq:estimatorbound}
 L_N\le E(\mathcal A)\le2U_N,
\end{equation}
where $L_N$ and $U_N$ are defined in \eqref{eq:FIRlower} and \eqref{eq:upperN}, respectively.
\end{corollary}
\begin{proof}
Theorems~\ref{thm:upper} and~\ref{thm:lower} give $L_N\le r(\Tset(0))\le U_N$. Lemma~\ref{lem:geometry} then completes the proof.
\end{proof}

\begin{remark}
The bounds show how noise changes the convergence of the identification error. When $\varepsilon=0$, $L_N=U_N=2K\rho^N/(1+\rho^{2N})$, so $E(\mathcal A)=\Theta(\rho^{-N})$: the worst-case error decreases exponentially to zero. For $0<\varepsilon<K$, the lower and upper bounds approach $\varepsilon$ and $2\varepsilon$, respectively, with
\[
 \varepsilon+\Theta(N^{-2})\le E(\mathcal A)
 \le2\varepsilon+\Theta(N^{-2}),\qquad N\to\infty.
\]
Thus, increasing the number of samples reduces the additional uncertainty due to finite frequency sampling, but cannot eliminate the worst-case error due to bounded noise. When $\varepsilon\ge K$, the bounds reduce to $K\le E(\mathcal A)\le2K$. In this regime, every system in the prior class is consistent with zero measurements, so increasing the number of samples cannot improve these worst-case bounds.
\end{remark}

\section{Stochastic IDENTIFICATION Bounds}\label{sec:stochastic}
We now ask whether bounded stochastic noise allows high probability diameter bounds, as in Definition~\ref{def:stochasticdiameter}, to fall below the noise level and converge to zero. We first recall an upper bound that holds for every admissible noise distribution.

Under Assumptions~\ref{ass:consistency}--\ref{ass:convexity}, every consistent algorithm satisfies, almost surely,
\begin{equation}\label{eq:stochasticzeroupper}
 \norm{\mathcal A(y)-H_0}_\infty
 \le e(\mathcal A,y)\le d(\Tset(y))\le2r(\Tset(0)),
\end{equation}
by \eqref{eq:localerror} and Lemma~\ref{lem:geometry}. Thus, for equally spaced samples, $2U_N$ is a valid high-probability diameter bound for every admissible noise distribution. When $0<\varepsilon<K$, it approaches $2\varepsilon$ at rate $O(N^{-2})$.
This quadratic convergence to $2\varepsilon$ raises the question: can a high-probability diameter bound converge to zero at the same rate? The following theorem shows that this is not possible in general. For uniform disk noise, every such bound must be at least of order $N^{-1}$, ruling out $O(N^{-2})$ convergence to zero.

\begin{theorem}\label{thm:stoch-lower}
Under \eqref{eq:prior}--\eqref{eq:data1}, let $K,\varepsilon>0$, $N\ge2$, and $0<\delta<1$, and use the sampling grid \eqref{eq:halfgrid}. Assume the noises at sampled frequencies in $[0,\pi]$ are independent. For $0<\theta_n<\pi$, let $\eta_n$ be uniform on the disk $\NoiseSet_\varepsilon$. At any sampled endpoint $\theta_n\in\{0,\pi\}$, let $\eta_n$ be uniform on the real interval $[-\varepsilon,\varepsilon]$. At the remaining frequencies, set $\eta_k=\overline{\eta_n}$ whenever $\theta_k=2\pi-\theta_n$. Let $p$ count the sampled frequencies in $(0,\pi)$ and $r$ those in $\{0,\pi\}$, so the full grid has $N=2p+r$ samples. Every high-probability diameter bound in Definition~\ref{def:stochasticdiameter} satisfies
\begin{equation}\label{eq:uniform_diameter}
 b_N(\delta)\ge2\min\{K,\varepsilon x_N(\delta)\},
\end{equation}
where $x_N(\delta)\in(0,1)$ is the unique solution of
\begin{equation}\label{eq:lens_threshold}
 \ell(x_N)^p(1-x_N)^r=\delta,
\end{equation}
with
\begin{equation}\label{eq:lens_fraction}
 \ell(x)=\frac2\pi\left(\arccos x-x\sqrt{1-x^2}\right),
 \qquad 0\le x\le1.
\end{equation}
For fixed $K,\varepsilon,\delta$, the right-hand side of \eqref{eq:uniform_diameter} satisfies
\begin{equation}\label{eq:uniform_diameter_rate}
 2\min\{K,\varepsilon x_N(\delta)\}
 =\frac{\pi\varepsilon\ln(1/\delta)}{N}+O(N^{-2}).
\end{equation}
Thus no high-probability diameter bound can be $o(N^{-1})$ under this noise model.
\end{theorem}
\begin{proof}
The proof is given in Appendix~\ref{app:stochastic-lower}.
\end{proof}

The proof of Theorem~\ref{thm:stoch-lower} explains the role of noise near the boundary. When $H_0=0$, the two constant systems $H_\pm=\pm a$, with $0<a\le\min\{K,\varepsilon\}$, both remain consistent whenever every noise sample satisfies $|\eta_n|\le\varepsilon-a$. Excluding either candidate therefore requires a noise sample close to the boundary, with a suitable phase. The function $\ell(a/\varepsilon)$ in \eqref{eq:lens_fraction} measures the probability that a disk sample leaves both candidates consistent. This motivates the boundary probability condition used in the following upper bound. Similar conditions occur in time-domain set-membership identification: \citet[Assumption~5]{MusaviEtAl2024} require positive probability arbitrarily close to the disturbance bounds, while \citet{XuLi2025} study probabilities near the boundary of general convex noise sets. Their results concern finite-dimensional parameter sets. Here we bound the $H_\infty$ diameter of an analytic class of transfer functions.

\subsection{A stochastic upper bound}\label{sec:stochastic_upper}
We now complement the lower bound in Theorem~\ref{thm:stoch-lower} with an upper bound on $d(\Tset(y))$ and study how its convergence depends on the noise distribution. To obtain this bound, we impose the following assumption, using the conjugate completion described in Section~\ref{sec:problem}.

\begin{assumption}\label{ass:noise}
The sampling grid is closed under conjugation. The noises at sampled frequencies in $[0,\pi]$ are independent and satisfy $|\eta_n|\le\varepsilon$ almost surely. Those in $(0,\pi)$ are identically distributed and circularly symmetric, with $\Pr_\eta\{|\eta_n|\ge\varepsilon-s\}>0$ for every $0<s<\varepsilon$. If $\theta_n\in\{0,\pi\}$, then $\eta_n$ is real-valued.
\end{assumption}

Write $\eta_n=r_ne^{j\phi_n}$, where $r_n=|\eta_n|$, and let $s_n=\varepsilon-r_n$ be the distance of the noise sample from its magnitude bound. For a sample with $0<\theta_n<\pi$, define
\begin{equation}\label{eq:q}
 q(s)=\Pr_\eta\{|\eta_n|\ge\varepsilon-s\},\qquad 0\le s\le\varepsilon.
\end{equation}
The same $q(s)$ applies to the conjugate samples. At every sampled frequency other than $0$ and $\pi$, circular symmetry gives, for $0\le s<\varepsilon$ and any phase arc $I$,
\begin{equation}\label{eq:phase}
 \Pr_\eta\{s_n\le s,\ \phi_n\in I\}=\frac{|I|}{2\pi}q(s).
\end{equation}
To express the dependence on the noise distribution, suppose $q(s)\ge c_\beta(s/\varepsilon)^\beta$ for $0<s\le\varepsilon$, with $\beta\ge0$ and $c_\beta\in(0,1]$. Related conditions on probabilities near the noise boundary are considered by \citet{XuLi2025}. Here $q(s)$ is the probability of lying within distance $s$ of the disk boundary. For uniform noise on $\NoiseSet_\varepsilon$, $q(s)=1-(1-s/\varepsilon)^2\ge s/\varepsilon$, so the bound holds with $\beta=1$. Noise of fixed magnitude $\varepsilon$ and uniform phase has $q(s)=1$, so the bound holds with $\beta=0$.

Using these boundary probabilities, the following theorem gives a high probability upper bound on the consistency set diameter, with convergence rates that depend on the noise distribution. Its proof follows the same idea as Theorem~\ref{thm:upper}: we use a second derivative bound to compare the maximum error with its value at a nearby sampled frequency. In the stochastic setting, we also need the noise at that sample to be close to the boundary and have a suitable phase. The probabilities above allow us to find such samples with high probability.
\begin{theorem}\label{teo:upper_stochastic}
Under Assumption~\ref{ass:noise} and \eqref{eq:prior}--\eqref{eq:gap}, let the samples be equally spaced, with $\Delta=2\pi/N$. Let $0<\delta<1$, $\Lambda=\ln(32N/\delta)$, and choose $s\in[0,\varepsilon)$ such that $q(s)>0$ and $N\ge2(3+32\Lambda/q(s))$. Then
\begin{equation}\label{eq:local-cor}
 b_N(\delta)=\min\{2U_N,\,2V_N(s)\}
\end{equation}
is a high-probability diameter bound, where
\begin{equation}\label{eq:VN}
 V_N(s)=\sqrt2\,s+\frac{32\pi^2K\rho}{(\rho-1)^2N^2}
 \left(3+\frac{32\Lambda}{q(s)}\right)^2,
\end{equation}
and $U_N$ is the worst-case bound given in \eqref{eq:upperN}.
In particular:
\begin{enumerate}
\item[(a)] If $q(0)=\Pr_\eta\{|\eta_n|=\varepsilon\}>0$, choosing $s=0$ gives a bound of order $O(N^{-2}\ln^2(N/\delta))$.
\item[(b)] If $q(s)\ge c_\beta(s/\varepsilon)^\beta$ for $0<s\le\varepsilon$, with $\beta>0$ and $c_\beta\in(0,1]$, choosing $s=\varepsilon(\Lambda/(c_\beta N))^{2/(1+2\beta)}$ gives, for all sufficiently large $N$, a bound of order $O((N^{-1}\ln(N/\delta))^{2/(1+2\beta)})$.
\end{enumerate}
\end{theorem}
\begin{proof}
The proof is given in Appendix~\ref{app:stochastic-upper}.
\end{proof}

\begin{remark}
For a fixed $s>0$, the stochastic term satisfies
\[
 2V_N(s)=2\sqrt2\,s+O\!\left(\frac{\ln^2(N/\delta)}{N^2}\right).
\]
Thus it approaches a positive tolerance with a quadratic rate up to a logarithmic factor. Assumption~\ref{ass:noise} allows $s$ to be chosen arbitrarily small, but a bound converging to zero requires $s$ to decrease with $N$, unless $q(0)>0$. The resulting rate depends on how quickly $q(s)$ decreases as $s\downarrow0$. For uniform disk noise, part~(b) gives $O((\ln(N/\delta)/N)^{2/3})$, while Theorem~\ref{thm:stoch-lower} gives a lower bound of $\Omega(N^{-1})$. Therefore, these bounds do not establish the exact stochastic convergence rate.
\end{remark}

\section{Numerical Examples}\label{sec:numerical}
We use two Monte Carlo experiments to examine how the mean identification error decreases as the number of frequency samples increases. The first considers errors above the noise bound $\varepsilon$ and compares their decay with the $O(N^{-2})$ sampling dependence of the deterministic bound in Theorem~\ref{thm:upper}. The second examines convergence below $\varepsilon$ and compares two noise distributions with the same magnitude bound, illustrating the dependence on the noise distribution described by Theorem~\ref{teo:upper_stochastic}. We model the semidefinite feasibility problems used in the identification process in YALMIP \citep{Lofberg2004}.

\subsection{Effect of the number of frequency samples}
\label{sec:experiment-fixed-system}
We consider a fixed fourth-order system obtained by rounding the transfer-function coefficients of a randomly generated stable system to two decimal places:
\begin{equation}\label{eq:experiment-plant}
H_0(z)=\frac{-2.81z+3.87z^2+0.92z^3-1.41z^4}
{1-0.77z-0.56z^2+0.38z^3+0.01z^4},
\end{equation}
with prior parameters $\rho=1.08$, $K=100$, and noise bound $\varepsilon=0.70$. For each $N\in\{4,8,12,16,24,32,\allowbreak40,48,56,64,80\}$, we use the equally spaced full-circle grid \eqref{eq:halfgrid}. Noise is uniform on the disk $|\eta|\le\varepsilon$, independently drawn at the $N/2$ positive-half-circle nodes and completed by conjugation. We repeat identification with $N_{\mathrm{exp}}=100$ independent noise realizations at each $N$. For each realization, we follow the consistency and interpolation approach of \citet[Section~IV]{ChenNettFan1995}. We first find corrections within the noise bound that make the measurements consistent with the prior, then construct a central Nevanlinna--Pick interpolant of the corrected data. Figure~\ref{fig:experiment-convergence} shows that the mean error closely follows the $N^{-2}$ reference over $N=32$--80, consistent with the quadratic sampling dependence of the bounds in Corollary~\ref{cor:error}.

\begin{figure}[t]
\centering
\includegraphics[width=\columnwidth]{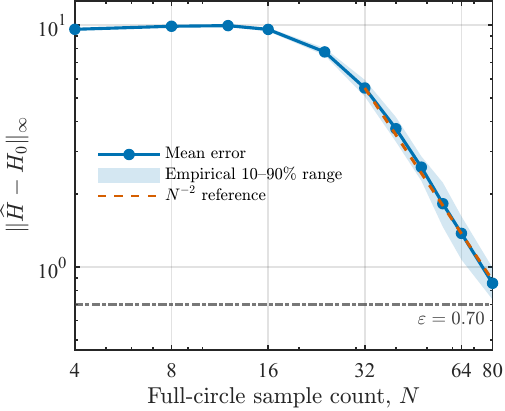}
\caption{Mean identification error versus $N$ over 100 noise realizations. Shading indicates the empirical 10th--90th percentile range. The dashed line is an $N^{-2}$ reference. The horizontal dash-dot line marks the noise bound $\varepsilon=0.70$.}
\label{fig:experiment-convergence}
\end{figure}

\subsection{Effect of the noise distribution}
\label{sec:experiment-noise-distribution}
To observe errors below the noise bound with at most 80 frequency samples, we use the following first-order system:
\[
H_0(z)=\frac{0.50z}{1-0.50z},
\]
with $\rho=1.25$, $K=2$, and $\varepsilon=0.50$. We compare noise uniform on the disk $|\eta|\le\varepsilon$ with noise of fixed magnitude $\varepsilon$ and uniform phase. In Theorem~\ref{teo:upper_stochastic}, these distributions correspond to $\beta=1$ and $\beta=0$, yielding upper-bound rates $O(N^{-2/3})$ and $O(N^{-2})$, respectively, omitting the logarithmic factors. For each $N\in\{8,16,24,32,40,48,64,80\}$, we use the grid \eqref{eq:halfgrid}, conjugate completion, and the same central Nevanlinna--Pick identification procedure as in Section~\ref{sec:experiment-fixed-system}. Figure~\ref{fig:experiment-noise-distribution} shows that the error for uniform disk noise closely follows the $N^{-2/3}$ reference. Fixed-magnitude noise gives faster decay, with the curve closer to the $N^{-2}$ reference at larger $N$.

\begin{figure}[t]
\centering
\includegraphics[width=\columnwidth]{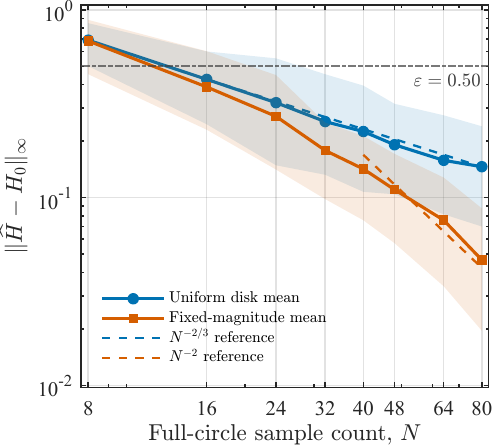}
\caption{Identification error over 100 noise realizations for each distribution. Shading indicates the empirical 10th--90th percentile range. Dashed lines show $N^{-2/3}$ and $N^{-2}$ reference rates, omitting the logarithmic factors in Theorem~\ref{teo:upper_stochastic}. The horizontal dash-dot line marks the noise bound $\varepsilon=0.50$.}
\label{fig:experiment-noise-distribution}
\end{figure}

\section{Conclusion}\label{sec:conclusion}
We studied the sample complexity of frequency domain identification under bounded noise. We improved previous worst-case bounds on the radius of the consistency set for zero measurements, establishing convergence to the noise level at the tight quadratic rate $\Theta(N^{-2})$. For bounded stochastic noise, the bounds can decrease below the noise level and converge to zero. Their convergence rate depends on how the noise distribution is concentrated near the boundary of the noise set. Future work will focus on sharpening the stochastic bounds to close the gap between the upper and lower rates, and on using the resulting uncertainty bounds in robust controller design.

\appendix
\section{Proof of Lemma~\ref{lem:curvature}}\label{app:proofs}
\begin{proof}
For a function $f$ analytic on $\D$ with $|f(w)|\le1$, Ruscheweyh's inequality, as stated in \citet[Theorem~1]{AndersonRovnyak2006}, gives
\begin{equation}\label{eq:ruscheweyh}
 \frac{|f^{(m)}(w)|}{m!}
 \le\frac{1-|f(w)|^2}
 {(1-|w|)^{m-1}(1-|w|^2)},\qquad m\ge1.
\end{equation}
Apply this to $f(w)=H(\rho w)/K$ at $w=z/\rho$, $|z|=1$. For $m=1,2$ it gives
\begin{align*}
 |H'(z)|&\le\frac{K\rho}{(\rho-1)(\rho+1)}
                  \left(1-\frac{|H(z)|^2}{K^2}\right),\\
 |H''(z)|&\le\frac{2K\rho}{(\rho-1)^2(\rho+1)}
                  \left(1-\frac{|H(z)|^2}{K^2}\right).
\end{align*}
Differentiating $F(\theta)=H(e^{j\theta})$ with respect to $\theta$, and writing $z=e^{j\theta}$, gives
\begin{align*}
 F'(\theta)&=jzH'(z),\\
 F''(\theta)&=-zH'(z)-z^2H''(z).
\end{align*}
The triangle inequality and the two derivative bounds give \eqref{eq:curvature}, since
\[
 \frac{K\rho}{(\rho-1)(\rho+1)}+
 \frac{2K\rho}{(\rho-1)^2(\rho+1)}
 =\frac{K\rho}{(\rho-1)^2}.
\]
\end{proof}

\section{Proof of Theorem \ref{thm:stoch-lower}}\label{app:stochastic-lower}
\begin{proof}
Take $H_0=0$, so $y_n=\eta_n$, and fix $0<a\le\min\{K,\varepsilon\}$. The constant systems $H_\pm(z)=\pm a$ belong to $\Sclass$. Both are consistent at sample $n$ exactly when
\[
 |\eta_n-a|\le\varepsilon,\qquad |\eta_n+a|\le\varepsilon.
\]
For a sample with $0<\theta_n<\pi$, this event is $\eta_n\in L_a$, where
\[
 L_a=(a+\Neps)\cap(-a+\Neps)
\]
is the lens formed by two disks of radius $\varepsilon$ with centers separated by $2a$. This lens lies inside $\Neps$. Since the noise is uniform on $\Neps$, its probability is the ratio of the lens area to $\pi\varepsilon^2$:
\[
 \Pr_\eta\{\eta_n\in L_a\}
 =\frac{2\varepsilon^2\arccos(a/\varepsilon)
       -2a\sqrt{\varepsilon^2-a^2}}{\pi\varepsilon^2}
 =\ell(a/\varepsilon).
\]
The lens is invariant under conjugation, so each conjugate pair contributes one independent event. At a real endpoint, both candidates are consistent exactly when $|\eta_n|\le\varepsilon-a$. Uniform noise on $[-\varepsilon,\varepsilon]$ gives this event probability $1-a/\varepsilon$. Independence before completion therefore yields
\begin{equation}\label{eq:uniform_lens_event}
 \Pr_\eta\{H_+,H_-\in\Tset(y)\}
 =\ell(a/\varepsilon)^p(1-a/\varepsilon)^r.
\end{equation}
On this event, $d(\Tset(y))\ge\norm{H_+-H_-}_\infty=2a$.

The function $\ell$ is continuous and strictly decreasing from $1$ to $0$ on $[0,1]$, since $\ell'(x)=-(4/\pi)\sqrt{1-x^2}<0$ for $0<x<1$. Thus \eqref{eq:lens_threshold} has a unique solution $x_N(\delta)\in(0,1)$. For every $0<a<\min\{K,\varepsilon x_N(\delta)\}$, the probability in \eqref{eq:uniform_lens_event} exceeds $\delta$. If $b_N(\delta)<2a$, then whenever both $H_+$ and $H_-$ are consistent, the diameter exceeds $b_N(\delta)$. This event has probability greater than $\delta$, violating the probability guarantee in Definition~\ref{def:stochasticdiameter}. Hence $b_N(\delta)\ge2a$. Since this holds for every $a<\min\{K,\varepsilon x_N(\delta)\}$, letting $a$ approach this minimum from below proves \eqref{eq:uniform_diameter}.

It remains to establish \eqref{eq:uniform_diameter_rate}. The expansion at $x=0$ gives
\[
 \ell(x)=1-\frac4\pi x+O(x^3),\qquad
 -\ln\ell(x)=\frac4\pi x+O(x^2).
\]
Because $r\le2$ and $p=(N-r)/2\to\infty$, \eqref{eq:lens_threshold} implies $x_N\to0$. Taking logarithms in that equation and using $-\ln(1-x)=x+O(x^2)$ yields
\begin{equation}\label{eq:lens_log_expansion}
 \ln(1/\delta)=\left(\frac{4p}{\pi}+r\right)x_N+O(Nx_N^2).
\end{equation}
Since $x_N\to0$ and $-\ln\ell(x)\ge(2/\pi)x$ for all sufficiently small positive $x$, the exact logarithmic form of \eqref{eq:lens_threshold} gives, for all sufficiently large $N$,
\[
 \ln(1/\delta)=-p\ln\ell(x_N)-r\ln(1-x_N)\ge\frac{2p}{\pi}x_N.
\]
Hence, using $2p=N-r$,
\[
 0<x_N\le\frac{\pi\ln(1/\delta)}{2p}
       =\frac{\pi\ln(1/\delta)}{N-r}.
\]
Therefore, $x_N=O(N^{-1})$. Substituting this estimate into \eqref{eq:lens_log_expansion} gives
\begin{equation}\label{eq:lens_threshold_asymptotic}
 x_N=\frac{\ln(1/\delta)}{4p/\pi+r}+O(N^{-2})
     =\frac{\pi\ln(1/\delta)}{2N}+O(N^{-2}).
\end{equation}
For fixed $K,\varepsilon>0$, we have $\varepsilon x_N<K$ for all sufficiently large $N$. Multiplying \eqref{eq:lens_threshold_asymptotic} by $2\varepsilon$ proves \eqref{eq:uniform_diameter_rate}.
\end{proof}

\section{Proof of Theorem \ref{teo:upper_stochastic}}\label{app:stochastic-upper}
We first show that the assumptions of Theorem~\ref{teo:upper_stochastic} ensure, with high probability, a nearby frequency grid point whose associated noise value is close to the boundary and has the required phase. All frequency and phase relations below are understood modulo $2\pi$.
\begin{lemma}\label{lem:covering}
Under the assumptions of Theorem~\ref{teo:upper_stochastic}, fix $0<\delta<1$ and $s$ as in that theorem, and set
\[
 \Lambda=\ln(32N/\delta),\qquad
 h=\frac{4\pi}{N}\left(3+\frac{32\Lambda}{q(s)}\right).
\]
At each sampled frequency $\theta_k$, write $\eta_k=(\varepsilon-s_k)e^{j\phi_k}$, where $s_k=\varepsilon-|\eta_k|$ is the distance to the boundary of the noise disk and $\phi_k$ is the noise phase.

With probability at least $1-\delta$, simultaneously for every frequency $\theta$, every angle $\alpha$, and every direction $\sigma\in\{-1,+1\}$, there exist an index $k$ and offsets $v,\psi$ satisfying:
\begin{itemize}
\item The sampled frequency is nearby: moving forward from $\theta$ around the frequency circle reaches $\theta_k$ within distance $h$,
\begin{equation}\label{eq:freqcover}
 \theta_k=\theta+v,\qquad 0\le v\le h.
\end{equation}
\item The noise is close to the boundary: its magnitude is within $s$ of $\varepsilon$,
\[
 s_k\le s,\qquad\text{equivalently}\qquad
 |\eta_k|\ge\varepsilon-s.
\]
\item The noise phase is close to the direction opposite to $\alpha$:
\begin{equation}\label{eq:phasecover}
 \phi_k=\alpha+\pi+\psi,\qquad
 0\le\sigma\psi\le\frac{\pi}{4}.
\end{equation}
Thus the phase differs from $\alpha+\pi$ by at most $\pi/4$, with $\sigma$ selecting the side of that direction.
\end{itemize}
The index $k$ may depend on $\theta$, $\alpha$, and $\sigma$. In simpler terms, with probability at least $1-\delta$, wherever we look in frequency and whichever noise direction we need, there is a nearby frequency grid point whose associated noise value is close to the boundary and has the required phase.
\end{lemma}
\begin{proof}
We use a finite covering and union bound argument, similar to the proof of Theorem~2.1 in \citet{ReznikovSaff2016}, adapted here to consecutive frequency grid points and noise phase intervals. Set $m=\lfloor Nh/(4\pi)\rfloor$. Using the definition of $h$ and the condition $N\ge2(3+32\Lambda/q(s))$ in Theorem~\ref{teo:upper_stochastic}, we obtain
\begin{equation}\label{eq:m-bound}
 m=\left\lfloor3+\frac{32\Lambda}{q(s)}\right\rfloor
 \ge2+\frac{32\Lambda}{q(s)},\qquad m\le N/2.
\end{equation}
For each starting index $i=0,\ldots,N-1$, define the block of indices $I_i=\{i,i+1,\ldots,i+m-1\}$, with all indices taken modulo $N$. Thus each of the $N$ blocks corresponds to $m$ consecutive frequency grid points $\theta_n$ and their associated noise values $\eta_n$, $n\in I_i$. To specify the noise phase intervals, introduce the reference angles $\alpha_\ell=\ell\pi/8$, $\ell=0,\ldots,15$. For each reference angle, consider the directed arcs
\[
 J_{\ell,\sigma}=\{\alpha_\ell+\pi+\sigma x:0\le x\le\pi/8\},
 \qquad\sigma\in\{-1,+1\}.
\]
Each arc has length $\pi/8$. Thus, for each index $n$ with $\theta_n\notin\{0,\pi\}$, circular symmetry and \eqref{eq:phase} give
\[
 \Pr_\eta\{s_n\le s,\ \phi_n\in J_{\ell,\sigma}\}=q(s)/16.
\]
Fix a block $I_i$ and an arc $J_{\ell,\sigma}$. Exclude indices with $\theta_n\in\{0,\pi\}$, and retain one predetermined index from each pair of conjugate frequency grid points represented in $I_i$. Since at most two indices are excluded at the endpoints and each pair contributes at most two indices, at least $(m-2)/2$ indices remain. Their associated noise variables are independent, and each satisfies $s_n\le s$ and $\phi_n\in J_{\ell,\sigma}$ with probability $q(s)/16$. If no index in $I_i$ satisfies both conditions, none of the retained indices does. Therefore,
\[
\begin{aligned}
 &\Pr_\eta\{\forall n\in I_i:\ s_n>s\ \text{or}\ \phi_n\notin J_{\ell,\sigma}\}
 \\
 &\qquad\le (1-q(s)/16)^{(m-2)/2}.
\end{aligned}
\]
By \eqref{eq:m-bound}, this upper bound is at most $e^{-(m-2)q(s)/32}\le e^{-\Lambda}$.
A union bound shows that, with probability at least $1-\delta$, every block contains a suitable noise sample for each of the specified phase intervals.

On this event, fix any frequency $\theta$, angle $\alpha$, and sign $\sigma$. Let $\theta_{n_0}$ be the first frequency grid point at or after $\theta$ in the forward direction around the circle, and choose the index block $I_{n_0}$. The grid point $\theta_{n_0}$ is at most $2\pi/N$ ahead of $\theta$. Every other grid point indexed by $I_{n_0}$ is at most $(m-1)2\pi/N$ ahead of $\theta_{n_0}$. Hence every $\theta_n$ with $n\in I_{n_0}$ has forward distance at most $2\pi m/N\le h/2$ from $\theta$. Choose the reference angle $\alpha_\ell$ at or immediately after $\alpha$ in the direction $\sigma$, so $\alpha_\ell=\alpha+\sigma w$ for some $0\le w<\pi/8$. For this block and this arc $J_{\ell,\sigma}$, there is an index $k\in I_{n_0}$ whose associated noise value satisfies $s_k\le s$ and $\phi_k\in J_{\ell,\sigma}$. By the definition of the arc,
\[
 \phi_k=\alpha_\ell+\pi+\sigma x
       =\alpha+\pi+\sigma(w+x),\qquad 0\le x\le\pi/8.
\]
Setting $\psi=\sigma(w+x)$ gives $0\le\sigma\psi\le\pi/4$, proving \eqref{eq:phasecover}. Since $k\in I_{n_0}$, the frequency grid point $\theta_k$ satisfies $\theta_k=\theta+v$ modulo $2\pi$ for a forward distance $0\le v\le2\pi m/N\le h/2\le h$. Together with $s_k\le s$, this proves \eqref{eq:freqcover} and completes the proof.
\end{proof}
\begin{proof}[Proof of Theorem~\ref{teo:upper_stochastic}]
We combine the Taylor argument from the proof of Theorem~\ref{thm:upper} with Lemma~\ref{lem:covering}, which provides a nearby sampled frequency whose noise is close to the boundary and has a suitable phase. Let $H\in\Tset(y)$ and define $G=H-H_0$. Since $H$ is consistent with the measurements and $y_k=H_0(e^{j\theta_k})+\eta_k$, we have, at each sampled frequency,
\[
|G(e^{j\theta_k})-\eta_k|
=|H(e^{j\theta_k})-y_k|\le\varepsilon.
\]

Write
\(
\eta_k=(\varepsilon-s_k)e^{j\phi_k}\) with
\(s_k = \varepsilon-|\eta_k|
\).
Projecting the inequality above onto the direction \(e^{j\phi_k}\) yields
\[
\Re\!\left(
e^{-j\phi_k}
\bigl(G(e^{j\theta_k})-\eta_k\bigr)
\right)
\ge -\varepsilon.
\]
Hence
\begin{equation}
\Re\!\left(e^{-j\phi_k}G(e^{j\theta_k})\right)\ge -s_k.
\label{eq:halfplane}
\end{equation}

Let \(F(\theta)= G(e^{j\theta}), \;
M = \|G\|_\infty\) and
choose \(\theta_*\) such that
\(
|F(\theta_*)|=M
\). Write
\(
F(\theta_*)=Me^{j\alpha}
\) and define
\[
g(v) =
\Re\!\left(e^{-j\alpha}F(\theta_*+v)\right),
\;
l(v)=
\Im\!\left(e^{-j\alpha}F(\theta_*+v)\right).
\]
Then \(
g(0)=M\) and \( l(0)=0\).
Moreover,
\[
g(v)\le |F(\theta_*+v)|\le M=g(0),
\]
so \(
g'(0)=0\). Let
\(
\lambda = l'(0)\). Since $G/2\in\Sclass$, Lemma~\ref{lem:curvature} gives
\[
 |F''(\theta)|\le\frac{2K\rho}{(\rho-1)^2}
 \left(1-\frac{|F(\theta)|^2}{4K^2}\right)
 \le \frac{2K\rho}{(\rho-1)^2}=L_2.
\]
Since taking real and imaginary parts cannot increase magnitude, this also gives $|g''(v)|\le L_2$ and $|l''(v)|\le L_2$. Applying Taylor's theorem to $g$, with $g(0)=M$ and $g'(0)=0$, yields
\begin{equation}
g(v)\ge M-\frac{L_2}{2}v^2.
\label{eq:g-taylor}
\end{equation}
Similarly, applying Taylor's theorem to $l$, with $l(0)=0$ and $l'(0)=\lambda$, gives
\begin{equation}
l(v)=\lambda v+r(v),
\qquad
|r(v)|\le \frac{L_2}{2}v^2.
\label{eq:k-taylor}
\end{equation}

Fix $\theta_*$ and $\alpha$ as above, and choose
\[
\sigma=
\begin{cases}
+1,&\lambda\ge0,\\
-1,&\lambda<0.
\end{cases}
\]
Applying Lemma~\ref{lem:covering} with $\theta=\theta_*$ and these values of $\alpha$ and $\sigma$ gives, with probability at least $1-\delta$, an index \(k\) for which the sampled frequency $\theta_k$, the distance $s_k$ to the boundary of the noise disk, and the noise phase $\phi_k$ satisfy
\[
\theta_k=\theta_*+v,\; 0\le v\le h, \; 
s_k\le s,
\]
and
\[
\phi_k=\alpha+\pi+\psi,
\qquad
0\le \sigma\psi\le \frac{\pi}{4}.
\]
At this frequency grid point, $e^{-j\alpha}F(\theta_k)=g(v)+j l(v)$ and $\phi_k=\alpha+\pi+\psi$. Therefore,
\[
\Re\!\left(e^{-j\phi_k}F(\theta_k)\right)
=-g(v)\cos\psi-l(v)\sin\psi.
\]
Substituting this identity into \eqref{eq:halfplane} and using $s_k\le s$ gives
\begin{equation}
g(v)\cos\psi+l(v)\sin\psi
\le s_k\le s.
\label{eq:keyineq}
\end{equation}

The choice of $\sigma$ ensures that $\lambda$ and $\psi$ have the same sign when they are nonzero, so $\lambda\sin\psi\ge0$. Since $v\ge0$, the linear term in \eqref{eq:k-taylor} is nonnegative after multiplication by $\sin\psi$. Thus
\[
l(v)\sin\psi=\lambda v\sin\psi+r(v)\sin\psi
\ge-\frac{L_2}{2}v^2|\sin\psi|.
\]
Also, $|\psi|\le\pi/4$ implies $\cos\psi>0$. We may therefore multiply \eqref{eq:g-taylor} by $\cos\psi$ without changing the inequality. Combining the two Taylor bounds with \eqref{eq:keyineq} yields
\[
s
\ge
\left(M-\frac{L_2}{2}v^2\right)\cos\psi
-
\frac{L_2}{2}v^2|\sin\psi|.
\]
Rearranging and dividing by $\cos\psi>0$, we obtain
\[
M
\le
\frac{s}{\cos\psi}
+
\frac{L_2}{2}v^2
\left(1+|\tan\psi|\right).
\]
Using $\cos\psi\ge1/\sqrt2$, $|\tan\psi|\le1$, and $v\le h$ now gives
\[
M\le \sqrt2\,s+L_2h^2.
\]
Substituting $L_2=2K\rho/(\rho-1)^2$ and the value of $h$ defined in Lemma~\ref{lem:covering}, the right-hand side equals $V_N(s)$ in \eqref{eq:VN}. Since $M=\norm{H-H_0}_\infty$ and the argument holds for every consistent $H$ on the same event, we conclude that, with probability at least $1-\delta$,
\[
\sup_{H\in\Tset(y)}
\|H-H_0\|_\infty
\le
V_N(s).
\]

To bound the diameter, take any $H_1,H_2\in\Tset(y)$. On this event, the triangle inequality gives
\[
 \norm{H_1-H_2}_\infty
 \le\norm{H_1-H_0}_\infty+\norm{H_2-H_0}_\infty
 \le2V_N(s).
\]
Taking the supremum over $H_1,H_2$ gives $d(\Tset(y))\le2V_N(s)$. Combining this with the deterministic bound $d(\Tset(y))\le2U_N$ proves \eqref{eq:local-cor}.

For part~(a), $q(0)>0$ allows us to set $s=0$. The first term of $V_N(s)$ then vanishes, and the remaining term is $O(\Lambda^2/N^2)$, giving the stated rate.

For part~(b), put $x_N=\Lambda/(c_\beta N)$ and choose $s=\varepsilon x_N^{2/(1+2\beta)}$. The assumed inequality $q(s)\ge c_\beta(s/\varepsilon)^\beta$ gives
\[
 \frac{\Lambda}{Nq(s)}
 \le x_N\left(\frac{\varepsilon}{s}\right)^\beta
 =x_N^{1/(1+2\beta)}.
\]
Consequently, $s=O(x_N^{2/(1+2\beta)})$ and
\[
\frac{1}{N^2}\left(3+\frac{32\Lambda}{q(s)}\right)^2
=\left(\frac3N+\frac{32\Lambda}{Nq(s)}\right)^2
=O\!\left(x_N^{2/(1+2\beta)}\right).
\]
Both terms of $V_N(s)$ therefore have the claimed upper bound. Finally, $x_N\to0$ gives $s<\varepsilon$ for all sufficiently large $N$, and
\[
\frac{2}{N}\left(3+\frac{32\Lambda}{q(s)}\right)
\le\frac6N+64x_N^{1/(1+2\beta)}\to0.
\]
This shows that the LHS is at most one for sufficiently large $N$, so $N\ge2(3+32\Lambda/q(s))$, as required. This proves part~(b).
\end{proof}

\section*{Acknowledgments}
The authors used Codex (OpenAI) and Claude (Anthropic) for writing and coding assistance. The authors reviewed, edited, and approved all changes and take full responsibility for the content.

\bibliography{ifacconf}

@article{ChenNettFan1995,
  author = {Chen, Jie and Nett, Carl N. and Fan, Michael K. H.},
  title = {Worst Case System Identification in {$H_\infty$}: Validation of a priori Information, Essentially Optimal Algorithms, and Error Bounds},
  journal = {IEEE Transactions on Automatic Control},
  year = {1995}, volume = {40}, number = {7}, pages = {1260--1265},
  doi = {10.1109/9.400481}
}

@article{ChenGu2002,
  author = {Chen, Jie and Gu, Guoxiang},
  title = {Worst-Case Asymptotic Properties of {$H_\infty$} Identification},
  journal = {IEEE Transactions on Circuits and Systems I: Fundamental Theory and Applications},
  year = {2002}, volume = {49}, number = {4}, pages = {437--446},
  doi = {10.1109/81.995658}
}

@article{AndersonRovnyak2006,
  author = {Anderson, J. M. and Rovnyak, J.},
  title = {On Generalized {Schwarz--Pick} Estimates},
  journal = {Mathematika},
  year = {2006}, volume = {53}, pages = {161--168},
  doi = {10.1112/S0025579300000085}
}

@book{SanchezPenaSznaier1998,
  author = {{S\'{a}nchez-Pe\~{n}a}, Ricardo S. and Sznaier, Mario},
  title = {Robust Systems Theory and Applications},
  publisher = {John Wiley \& Sons},
  year = {1998}, isbn = {9780471176275}
}

@article{HelmickiJacobsonNett1991,
  author = {Helmicki, Arthur J. and Jacobson, Clas A. and Nett, Carl N.},
  title = {Control oriented system identification: a worst-case/\allowbreak deterministic approach in {$H_\infty$}},
  journal = {IEEE Transactions on Automatic Control},
  year = {1991}, volume = {36}, number = {10}, pages = {1163--1176},
  doi = {10.1109/9.90229}
}

@article{AkcayGuKhargonekar1994,
  author = {Ak\c{c}ay, H\"{u}seyin and Gu, Guoxiang and Khargonekar, Pramod P.},
  title = {Identification in {$H_\infty$} with nonuniformly spaced frequency response measurements},
  journal = {International Journal of Robust and Nonlinear Control},
  year = {1994}, volume = {4}, number = {4}, pages = {613--629},
  doi = {10.1002/rnc.4590040412}
}

@article{ParriloEtAl1998,
  author = {Parrilo, Pablo A. and Sznaier, Mario and {S\'{a}nchez Pe\~{n}a}, Ricardo S. and Inanc, T.},
  title = {Mixed Time/Frequency-Domain Based Robust Identification},
  journal = {Automatica},
  year = {1998}, volume = {34}, number = {11}, pages = {1375--1389},
  doi = {10.1016/S0005-1098(98)00083-1}
}

@article{MilaneseTaragna2005,
  author = {Milanese, Mario and Taragna, Michele},
  title = {{$H_\infty$} set membership identification: A survey},
  journal = {Automatica},
  year = {2005}, volume = {41}, number = {12}, pages = {2019--2032},
  doi = {10.1016/j.automatica.2005.07.007}
}

@inproceedings{XuLi2025,
  author = {Xu, Haonan and Li, Yingying},
  title = {On the Sample Complexity of Set Membership Estimation for Linear Systems with Disturbances Bounded by Convex Sets},
  booktitle = {2025 American Control Conference},
  year = {2025}, pages = {3856--3862},
  doi = {10.23919/ACC63710.2025.11107618}
}

@article{ZengYuOzay2025,
  author = {Zeng, Xiong and Yu, Jing and Ozay, Necmiye},
  title = {System Identification Under Bounded Noise: Optimal Rates Beyond Least Squares},
  journal = {IEEE Control Systems Letters},
  year = {2025}, volume = {9}, pages = {1087--1092},
  doi = {10.1109/LCSYS.2025.3580036}
}

@article{CeroneFossonPirreraRegruto2025,
author = {Cerone, V. and Fosson, S.M. and Pirrera, S. and Regruto, D.},
title={Set-membership identification of continuous-time {MIMO} systems via {Tustin} discretization},
journal = {{arXiv preprint}, arXiv:2508.19348},
year={2025}}

@inproceedings{MusaviEtAl2024,
  author = {Musavi, Negin and Guo, Ziyao and Dullerud, Geir and Li, Yingying},
  title = {Identification of Analytic Nonlinear Dynamical Systems with Non-asymptotic Guarantees},
  booktitle = {Advances in Neural Information Processing Systems},
  volume = {37},
  year = {2024},
  doi = {10.52202/079017-2715}
}

@article{Ljung2010,
  author = {Ljung, Lennart},
  title = {Perspectives on system identification},
  journal = {Annual Reviews in Control},
  year = {2010}, volume = {34}, number = {1}, pages = {1--12},
  doi = {10.1016/j.arcontrol.2009.12.001}
}

@article{TsiamisEtAl2023,
  author = {Tsiamis, Anastasios and Ziemann, Ingvar and Matni, Nikolai and Pappas, George J.},
  title = {Statistical Learning Theory for Control: A Finite-Sample Perspective},
  journal = {IEEE Control Systems Magazine},
  year = {2023}, volume = {43}, number = {6}, pages = {67--97},
  doi = {10.1109/MCS.2023.3310345}
}

@article{MilaneseVicino1991,
  author = {Milanese, Mario and Vicino, Antonio},
  title = {Optimal estimation theory for dynamic systems with set membership uncertainty: An overview},
  journal = {Automatica},
  year = {1991}, volume = {27}, number = {6}, pages = {997--1009},
  doi = {10.1016/0005-1098(91)90134-N}
}

@inproceedings{SchoukensPintelonRolain2004,
  author = {Schoukens, J. and Pintelon, Rik and Rolain, Yves},
  title = {Time domain identification, frequency domain identification. Equivalencies! Differences?},
  booktitle = {Proceedings of the 2004 American Control Conference},
  year = {2004}, volume = {1}, pages = {661--666},
  url = {https://skoge.folk.ntnu.no/prost/proceedings/acc04/Papers/0113_WeM01.2.pdf}
}

@article{HonarpishehSznaier2025,
  author = {Honarpisheh, Arya and Sznaier, Mario},
  title = {Frequency Response Identification of Low-Order Systems: Finite-Sample Analysis},
  journal = {arXiv preprint arXiv:2508.17142},
  year = {2025},
  doi = {10.48550/arXiv.2508.17142},
  url = {https://arxiv.org/abs/2508.17142}
}

@article{ReznikovSaff2016,
  author = {Reznikov, A. and Saff, E. B.},
  title = {The Covering Radius of Randomly Distributed Points on a Manifold},
  journal = {International Mathematics Research Notices},
  year = {2016},
  volume = {2016},
  number = {19},
  pages = {6065--6094},
  doi = {10.1093/imrn/rnv342}
}

@article{SharifGharaviAlkhansariKhalaj2003,
  author = {Sharif, Masoud and Gharavi-Alkhansari, Mohammad and Khalaj, Babak H.},
  title = {On the Peak-to-Average Power of {OFDM} Signals Based on Oversampling},
  journal = {IEEE Transactions on Communications},
  year = {2003},
  volume = {51},
  number = {1},
  pages = {72--78},
  doi = {10.1109/TCOMM.2002.807619}
}

@inproceedings{Lofberg2004,
  author = {L{\"o}fberg, Johan},
  title = {{YALMIP}: A Toolbox for Modeling and Optimization in {MATLAB}},
  booktitle = {Proceedings of the CACSD Conference},
  address = {Taipei, Taiwan},
  year = {2004}
}
\end{document}